\documentclass[12pt,a4paper]{article}

\usepackage{amssymb,amsmath}
\usepackage{bbm}
\usepackage{stmaryrd}
\usepackage[mathcal]{euscript}

\numberwithin{equation}{section}

\newcounter{thMM}
\newcounter{leMM}
\newcounter{deFF}
\newcounter{exMP}
\newcounter{prOP}
\newcounter{coRR}
\newenvironment{theorem}[1][Theorem]{\refstepcounter{thMM}\trivlist
   \item[\hskip19pt{\bf #1~\arabic{thMM}.}]\it\hskip3pt}{\endtrivlist}
\newenvironment{lemma}[1][Lemma]{\refstepcounter{leMM}\trivlist
   \item[\hskip19pt{\bf #1~\arabic{leMM}.}]\it\hskip3pt}{\endtrivlist}
\newenvironment{definition}[1][Definition]{\refstepcounter{deFF}\trivlist
   \item[\hskip19pt{\bf #1~\arabic{deFF}.}]\rm\hskip3pt}{\endtrivlist}
\newenvironment{example}[1][Example]{\refstepcounter{exMP}\trivlist
   \item[\hskip19pt{\bf #1~\arabic{exMP}.}]\rm\hskip3pt}{\endtrivlist}
\newenvironment{proposition}[1][Proposition]{\refstepcounter{prOP}\trivlist
   \item[\hskip19pt{\bf #1~\arabic{prOP}.}]\it\hskip3pt}{\endtrivlist}
\newenvironment{corollary}[1][Corollary]{\refstepcounter{coRR}\trivlist
   \item[\hskip19pt{\bf #1~\arabic{coRR}.}]\it\hskip3pt}{\endtrivlist}

\newenvironment{proof}[1][Proof]{\begin{trivlist}
\item[\hskip \labelsep {\bfseries #1}] }{ \begin{flushright}$\square$\end{flushright}\end{trivlist}}
\newenvironment{remark}[1][Remark]{\begin{trivlist}
\item[\hskip \labelsep {\bfseries #1}]}{\end{trivlist}}

\newcommand{\Dleft}{[\hspace{-1.5pt}[}
\newcommand{\Dright}{]\hspace{-1.5pt}]}
\newcommand{\SN}[1]{\Dleft #1 \Dright}

\newcommand{\Q}{\mathcal{Q}}

\newcommand{\D}{\mathbbmss{D}}

\DeclareMathOperator{\Vect}{Vect}

\DeclareMathOperator{\w}{w}
\DeclareMathOperator{\proj}{proj}

\begin{document}
\bibliographystyle{plain}

\author{Andrew James Bruce\\ \small \emph{Pembrokeshire College},\\
\small \emph{Haverfordwest, Pembrokeshire},\\\small  \emph{SA61 1SZ, UK}\\\small email:\texttt{andrewjamesbruce@googlemail.com}}
\date{\today}
\title{Odd Jacobi manifolds: general theory and applications to  generalised Lie algebroids}
\maketitle

\begin{abstract}
In this paper we define  a Grassmann odd analogue of  Jacobi structure on a  supermanifold. The basic properties are explored. The  construction of odd Jacobi manifolds is then used to reexamine the notion of a Jacobi algebroid. It is shown that Jacobi algebroids can be understood in terms of a kind of \emph{curved} Q-manifold, which we will refer to as a quasi Q-manifold.
\end{abstract}

\begin{small}
\end{small}

\tableofcontents

\section{Introduction}\label{sec:Introduction}
Lichnerowicz \cite{Lichnerowicz1977} introduced the notion of a Poisson manifold as well as a Jacobi manifold. Such  manifolds have found applications in classical mechanics and play an important role in quantisation. In this paper we introduce the notion of  (Grassmann) odd Jacobi (super)manifolds. The main motivation for doing so is the fact that odd Poisson brackets, also known as Schouten brackets or in the physics literature as antibrackets have found  important application in the Batalin--Vilkoviski antifield formalism of gauge theories \cite{Batalin:1981jr,Batalin:1984jr}.\\

Let us briefly recall the definition of a (classical or even) Jacobi manifold\footnote{This can all be generalised to supermanifolds, however doing so is inessential for this paper.} as a manifold equipped with a bi-vector $\Gamma$ and a vector field $E$ such that

\begin{equation}\label{eqn:evenjacobi}
\SN{\Gamma, \Gamma} =  2 E \wedge \Gamma  \hspace{30pt} L_{E}\Gamma = \SN{E,\Gamma} =0,
\end{equation}

where $\SN{\bullet, \bullet}$ is the Schouten--Nijenhuis bracket. (There is some leeway here in  signs due to  conventions). Given these structures one builds a Lie algebra on $C^{\infty}(M)$ viz

\begin{equation}
\nonumber \{ f,g \}_{J} = \Gamma(df, dg) + fE(g) - E(f)g.
\end{equation}

Importantly this  Jacobi bracket is not a derivation unless $E =0$, which reduces the theory to that of Poisson geometry. That is if $E \neq 0$ the Leibniz rule is not exactly satisfied but contains a correction proportional to $\{f, \mathbbmss{1}  \}_{J}$.   For a modern review of Jacobi structures, including $\mathbbmss{Z}$-graded versions see \cite{Grabowski2001,Grabowski2003}.\\

In  this paper we  construct a Grassmann odd analogue of Jacobi structures on supermanifolds\footnote{We refrain from calling odd Jacobi structures \emph{Schouten--Jacobi structures} as this name has been already allocated in \cite{Grabowski2001} to mean something specific and not identical to constructions found in this work.}.  \\

Studying  geometric structures on supermanifolds provides a method of confirming the philosophy that \emph{ supermanifolds can  informally be thought of as ``manifolds" with commuting and anticommuting coordinates}. More than this, the inclusion of Grassmann odd degrees of freedom allow one to construct geometric structures that have no classical analogue on manifolds. These ``odd structures" such are of mathematical interest and often find applications in physics. For example, odd symplectic and Schouten structures are at the heart of the  Batalin--Vilkoviski antifield formalism.\\

The theory of Schouten manifolds and Q-manifolds (supermanifolds with a homological vector field) are both well established within the mathematical physics literature. However, construction of an odd analogue of classical Jacobi structures appears to have been overlooked. The expected properties of such a construction cannot be taken for granted.  In this work we preform  an  examination of  the basic properties of odd Jacobi structures.  \\

 The theory of odd Jacobi structures is described by ``almost Schouten structure" $S \in C^{\infty}(T^{*}M)$, that is a Grassmann odd  fibre-wise polynomial of degree two and a homological vector field $Q \in \Vect(M)$, together with natural conditions analogous to Eqn.(\ref{eqn:evenjacobi}). The associated odd Jacobi brackets   satisfy the standard properties of  Schouten  brackets, with the exception of the derivation property.\\

 A large proportion  of this paper is devoted to making the previous statement concrete and exploring the elementary properties of supermanifolds equipped with odd Jacobi structures.  We then proceed to use this technology to reexamine Jacobi algebroids. \\

Jacobi algebroids were first introduced by Iglesias and  Marrero \cite{Iglesias2001} under the name of generalised Lie algebroids. Such structures were the \emph{revisited} by Grabowski and Marmo \cite{Grabowski2001}. Jacobi algebroids represent a nice generalistion of the concept of a Lie algebroid. Indeed, Jacobi algebroids can be understood as Lie algebroids in the presence of a 1-cocycle \cite{Iglesias2001}. Recall the notion of a Lie algebroid as a vector bundle $E \rightarrow M$ equipped with a Lie bracket on the sections $[\bullet, \bullet]: \Gamma(E) \otimes \Gamma(E) \rightarrow \Gamma(E)$ together with an anchor $a: \Gamma(E) \rightarrow \Gamma(TM)$ that satisfy the Leibniz rule

\begin{equation}
\nonumber [u,fv] = a(u)[f] \: v  +  (-1)^{\widetilde{u}\widetilde{f}} f [u,v],
\end{equation}

where $u,v \in \Gamma(E)$ and $f \in C^{\infty}(M)$. The Leibniz rule implies that the anchor is actually a Lie algebra morphism: $a\left([u,v]\right) = [a(u), a(v)]$.  A Lie algebroid  can also  be understood in terms of:

\begin{enumerate}
\item a homological vector field of weight one on the total space of $\Pi E$. \label{homvect}
\item a weight minus one Poisson structure on the total space of $E^{*}$.
\item a weight minus one Schouten structure on the total space of $\Pi E^{*}$. \label{schouten}
\end{enumerate}

Here $\Pi$ is the parity reversion functor which shifts the Grassmann parity of the fibre coordinates, but does not effect the base coordinates. Also the parity reversion functor does not effect the assignment of the weights.\\

At this point we must remark that we will be working the the category of graded manifolds when discussing algebroids. That is  we will be working with supermanifolds equipped with a privileged class of atlases where the coordinates are assigned weights taking values in $\mathbbmss{Z}$ and the coordinate transformations are polynomial in coordinates with nonzero weights respecting the  weight. Generally the weight  will be independent of the Grassmann parity. Moreover, any sign factors that arise will be due to the Grassmann parity and we do not include any possible extra signs due to the weight. In simpler terms, we have a manifold equipped with a  distinguished class of charts and diffeomorphisms between them respecting the $\mathbbmss{Z}_{2}$-grading as well as the additional $\mathbbmss{Z}$-grading. These gradings then pass over to geometric objects (tensor and tensor-like objects) on graded manifolds. For further details about graded manifolds one can consult  \cite{Grabowski2009,Roytenberg:2001,Voronov:2001qf}. \\

Jacobi algebroids c.f. \cite{Grabowski2001,Iglesias2001} are understood in terms of a weight minus one even Jacobi bracket on $C^{\infty}(E^{*})$ for a given vector bundle $E \rightarrow M$. By  an even Jacobi bracket we mean a Poisson-like bracket in which the Leibnitz rule is weakened in  very specific way. We say more about this shortly.\\

In this paper we address the ``odd" approach to Jacobi algebroids generalising \ref{homvect}. and \ref{schouten}. of the above list mimicking the odd-super constructions related to Lie algebroids.  We shall show that Jacobi algebroids can be described in terms of an odd vector field  $\D \in \Vect(\Pi E)$ of weight one that satisfies the condition that $[\D, \D] = q \: \D$ for some odd zero form $q \in C^{\infty}(\Pi E)$. \\

The description of Lie algebroids in terms of homological vector fields is due to Va$\breve{\textrm{{\i}}}$ntrob \cite{Vaintrob:1997}. The deep links between Poisson geometry and Lie algebroids can be traced back to Coste, Dazord \& Weinstein \cite{Coste1987}. \\

\begin{remark}
Antunes and Laurent--Gengoux \cite{Antunes2011} studied Jacobi structures using the supergeometric formulation of \emph{Buttin's big bracket}. They use this formulism to efficiently describe Jacobi algebroids and Jacobi bialgebroids. The supergeometric approach of the big bracket is different, but certainly related to that presented here. Grabowski  \cite{Grabowski2011}, just after this work was completed, provided a systematic and quite general approach to contact and Jacobi structures on graded supermanifolds. Also note the work of Mehta \cite{Mehta2011}, which we will comment on later.
\end{remark}

\noindent \textbf{Preliminaries} \\
All vector spaces and algebras will be $\mathbbmss{Z}_{2}$-graded.   We will generally  omit the prefix \emph{super}. By \emph{manifold} we will mean a \emph{smooth supermanifold}. We denote the Grassmann parity of an object by \emph{tilde}: $\widetilde{A} \in \mathbbmss{Z}_{2}$. By \emph{even} or \emph{odd} we will be referring explicitly to the Grassmann parity.\\

 A \emph{Poisson} $(\varepsilon = 0)$  or \emph{Schouten} $(\varepsilon = 1)$ \emph{algebra} is understood as a vector space $A$ with a bilinear associative multiplication and a bilinear operation $\{\bullet , \bullet\}: A \otimes A \rightarrow A$ such that:
\begin{list}{}
\item \textbf{Grading} $\widetilde{\{a,b \}_{\varepsilon}} = \widetilde{a} + \widetilde{b} + \varepsilon$
\item \textbf{Skewsymmetry} $\{a,b\}_{\varepsilon} = -(-1)^{(\tilde{a}+ \varepsilon)(\tilde{b}+ \varepsilon)} \{b,a \}_{\varepsilon}$
\item \textbf{Jacobi Identity} $\displaystyle\sum\limits_{\textnormal{cyclic}(a,b,c)} (-1)^{(\tilde{a}+ \varepsilon)(\tilde{c}+ \varepsilon)}\{a,\{b,c\}_{\varepsilon}  \}_{\varepsilon}= 0$
\item \textbf{Leibniz Rule} $\{a,bc \}_{\varepsilon} = \{a,b \}_{\varepsilon}c + (-1)^{(\tilde{a} + \varepsilon)\tilde{b}} b \{a,c \}_{\varepsilon}$
\end{list} \vspace{10pt}
for all homogenous elements $a,b,c \in A$.\\

If the Leibniz rule does not hold identically, but is modified as
\begin{equation}
\{a,bc \}_{\varepsilon} = \{a,b \}_{\varepsilon}c + (-1)^{(\tilde{a} + \varepsilon)\tilde{b}} b \{a,c \}_{\varepsilon} - \{a ,\mathbbmss{1}  \} bc,
\end{equation}

then we have even ($\epsilon = 0)$ or odd ($\epsilon = 1)$ \emph{Jacobi algebras}.\\

A manifold $M$ such that $C^{\infty}(M)$ is a Poisson/Schouten algebra is known as a \emph{Poisson/Schouten manifold}. In particular the cotangent of a manifold comes equipped with a canonical Poisson structure.\\

Let us employ   natural local coordinates $(x^{A}, p_{A})$ on $T^{*}M$, with $\widetilde{x}^{A} = \widetilde{A}$ and $\widetilde{p}_{A} = \widetilde{A}$. Local diffeomorphisms on $M$ induce vector  bundle automorphism on $T^{*}M$ of the form
\begin{equation}
\nonumber \overline{x}^{A} = \overline{x}^{A}(x), \hspace{30pt} \overline{p}_{A}  = \left(\frac{\partial x^{B}}{\partial \overline{x}^{A}}\right)p_{B}.
\end{equation}

We will in effect use the local description as a \emph{natural vector bundle} to define the cotangent bundle of a supermanifold.  The canonical Poisson bracket on the cotangent is given by

\begin{equation}
\{ F,G \} = (-1)^{\widetilde{A} \widetilde{F} + \widetilde{A}} \frac{\partial F}{\partial p_{A}}\frac{\partial G}{\partial x^{A}} - (-1)^{\widetilde{A}\widetilde{F}}\frac{\partial  F}{\partial x^{A}} \frac{\partial G}{\partial p_{A}}.
\end{equation}\\

A manifold equipped with an odd vector field $Q$, such that the non-trivial condition $Q^{2}= \frac{1}{2}[Q,Q]=0$ holds, is known as a \emph{Q-manifold}. The vector field $Q$ is known as a \emph{homological vector field} for obvious reasons. \\

The structure of a Schouten manifold is encoded in an ``odd Hamiltonian" $\bar{S} \in C^{\infty}(T^{*}M)$ which Poisson self-commutes, $\{\bar{S}, \bar{S} \}=0$. The associated Schouten bracket is given by
\begin{equation}
\nonumber \SN{f,g} = (-1)^{\widetilde{f}+1}\{ \{ \bar{S}, f \},g \},
\end{equation}

with $f,g \in C^{\infty}(M)$.\\

A Schouten manifold equipped with a homological vector field such that $L_{\bar{Q}}\bar{S} =0$ is known as a \emph{QS-manifold}, see \cite{Voronov:2001qf} and for the odd symplectic case see \cite{Alexandrov:1995kv}. Importantly the homological vector field is a derivation over the Schouten bracket.\\

\begin{definition}\label{def:quasiQ}
The triple $(M, \D, q)$ with  $M$ a manifold,  $\D \in \Vect(M)$ an odd vector field and  $q \in C^{\infty}(M)$ an odd function such that:

\begin{equation}
  \D^{2} = \frac{1}{2}[\D, \D] =  q \D, \hspace{15pt}\textnormal{and}\hspace{15pt}  \D[q] =0,
\end{equation}

shall be called a \textbf{quasi Q-manifold}. The vector field $\D$ shall be known as an \textbf{almost homological vector field}. The odd function $q$ shall be known as the \textbf{curving function}.
\end{definition}

\begin{definition}
Let $(M, \D_{M}, q_{M})$ and $(N, \D_{N}, q_{N})$ be quasi Q-manifolds and let $\phi: M \rightarrow N$ be  a smooth map. Then $\phi$ is said to be a \textbf{morphism of quasi Q-manifolds} if and only if
\begin{enumerate}
\item $\D_{M}\left( \phi^{*}f \right) = \phi^{*}\left(\D_{N}f \right)$ for all $f \in C^{\infty}(N)$. That is the almost homological vector fields are $\phi$-related.
\item $\phi^{*}q_{N} = q_{M}$. That is the curving functions match.
\end{enumerate}
\end{definition}

Quasi Q-manifolds and their morphisms form a category. Also note that if $q=0$ then we have the category of \emph{Q-manifolds} and $\D$ is a \emph{homological vector field}, \cite{Alexandrov:1995kv}. The curving functions ``measure" the failure of the homological condition of $\D$ and thus represent a kind of ``curvature". The other extreme is to set $\D =0$ and then keep $q$ as some distinguished odd function. For example, one could consider \emph{(higher) Schouten manifolds} as examples of quasi Q-manifolds. The far extreme is the trivial structure  of $\D =0$ and $q=0$ and we recover the full category of supermanifolds.\\

\section{General Theory}\label{sect:generaltheory}
\subsection{Odd Jacobi structures}\label{subsect:oddjacobi}
In this section we define odd Jacobi manifolds and show that much of the theory of classical Jacobi manifolds carries over to the odd case. One should of course keep in mind the similarities and difference with Schouten manifolds. To some extent many of the results here are known to experts, though they appear not to have been written down in one place. \\
\begin{definition}
An \textbf{odd Jacobi structure} $(S,Q)$ on a manifold $M$  consists of
\begin{itemize}
\item an odd function $S \in C^{\infty}(T^{*}M)$, of degree two in fibre coordinates,
\item an odd vector field $Q \in \Vect(M)$,
\end{itemize}
such that the following conditions hold:
\begin{enumerate}
\item the homological condition $Q^{2} = \frac{1}{2} [Q,Q]=0$,
\item the invariance condition  $L_{Q}S = 0$,
\item the compatibility condition $\{S,S \}= - 2 \Q S $,
\end{enumerate}
 Here $\Q \in C^{\infty}(T^{*}M)$ is the principle symbol or ``Hamiltonian"  of the vector field $Q$. The brackets $\{ \bullet, \bullet \}$ are the canonical Poisson brackets on the cotangent bundle of the manifold.
\end{definition}

\begin{remark}
Note  $[E,E]=0$ automatically for the even vector field $E$ in the definition of an even or classical Jacobi structure. For odd structures this is a non-trivial condition.  Specifically, the underlying manifold $M$ is in fact a Q-manifold in the odd case.
\end{remark}
Note that the above conditions 1. and 2.  can be written entirely in terms of $S$ and $\Q$ and the canonical Poisson bracket as
\renewcommand{\labelenumi}{$\arabic{enumi}^{\prime}$.}
\begin{enumerate}
\item $\{\Q, \Q \} =0$,
\item $\{ \Q, S \}=0$,
\end{enumerate}
\renewcommand{\labelenumi}{\arabic{enumi}.}
which will be very convenient for calculational purposes.\\

\begin{remark}
Note that a QS-manifold can also be understood in a very similar way by taking the principle symbol of the homological vector field. In particular we have the triple $(M, \bar{S}, \bar{\Q})$, such that \\

\begin{tabular}{lll}
$\{ \bar{S}, \bar{S}\} =0$,&
$ \{\bar{\Q}, \bar{S} \}=0$,&
$ \{ \bar{Q}, \bar{Q} \} =0$.
\end{tabular}
\end{remark}

In natural local coordinates $(x^{A}, p_{A})$ on $T^{*}M$ the odd Jacobi structure is given by

\begin{equation}
 S = \frac{1}{2!}S^{AB}(x) p_{B}p_{A},\hspace{15pt} \textnormal{and}  \hspace{15pt} \Q = Q^{A}(x)p_{A},
\end{equation}

 the homological vector field is in local coordinates given by $Q = Q^{A}\frac{\partial}{\partial x^{A}}$. In local coordinates the conditions on the structures can be written as

\begin{eqnarray}
\nonumber \{ \Q, \Q \} &=& 2 Q^{B}\frac{\partial Q^{A}}{\partial x^{B}} p_{A} = 0,\\
\nonumber \{\Q , S  \} &=& \left(\frac{1}{2} Q^{C}\frac{\partial S^{BA}}{\partial x^{D}} + (-1)^{\widetilde{B}} S^{BC}\frac{\partial Q^{A}}{\partial x^{c}}  \right)p_{A}p_{B}=0,\\
\nonumber \{ S,S \}+ 2 \Q S &=& (-1)^{\widetilde{C}}\left( S^{CD} \frac{\partial S^{BA}}{\partial x^{D}} + Q^{C}S^{BA}  \right)p_{A}p_{B}p_{C} = 0.
\end{eqnarray}

\begin{definition}
A manifold equipped with an odd Jacobi structure $(S,Q)$ shall be known as an \textbf{odd Jacobi manifold}.
\end{definition}

As we shall see, the algebra of smooth functions $C^{\infty}(M)$ of an odd Jacobi manifold is in fact an odd Jacobi algebra. Following the natural analogue of Lichnerowicz's constructions we have the following definition:

\begin{definition}
The \textbf{odd Jacobi bracket} on $C^{\infty}(M)$ is defined as
\begin{eqnarray}
\SN{f,g}_{J} &=&  (-1)^{\widetilde{f}+1} \{ \{ S,f \},g    \} - (-1)^{\widetilde{f}+1} \{ \Q, fg \}\\
\nonumber &=&(-1)^{(\widetilde{B}+1)\widetilde{f}  +1} S^{BA} \frac{\partial f}{\partial x^{A}} \frac{\partial g}{\partial x^{B}} + (-1)^{\widetilde{f}} \left(Q^{A} \frac{\partial f}{\partial x^{A}}   \right)g  + f \left( Q^{A}\frac{\partial g}{\partial x^{A}}  \right),
\end{eqnarray}
with $f,g \in C^{\infty}(M)$.
\end{definition}

\begin{theorem}\label{theorem:oddjacobialgebra}
The odd Jacobi bracket defines an odd Jacobi algebra on $C^{\infty}(M)$. That is  the odd Jacobi bracket has the following properties:
\begin{enumerate}
\item Symmetry:  $\SN{f,g}_{J} = - (-1)^{(\widetilde{f}+1)(\widetilde{g}+1)} \SN{g,f}_{J}$.
\item Jacobi identity:  $\displaystyle\sum\limits_{\textnormal{cyclic}(f,g,h)} (-1)^{(\widetilde{f}+1)(\widetilde{h}+1)} \SN{f, \SN{g,h}_{J}}_{J}=0$.\\
\item Generalised Leibniz rule: $\SN{f, gh}_{J} = \SN{f,g}_{J}h + (-1)^{(\widetilde{f}+1) \widetilde{g}} g \SN{f,h}_{J} - \SN{f, \mathbbmss{1}}_{J}gh$.
\end{enumerate}
\end{theorem}

\begin{proof} We proceed to prove the above theorem by making use of the local descriptions.\\
\begin{enumerate}
\item The symmetry is clear from the definition given that $S^{AB} = (-1)^{\widetilde{A}\widetilde{B}}S^{BA}$.
\item As the odd Jacobi bracket is odd and skew-antisymmetric it is sufficient to examine the \emph{even diagonal} in proving the Jacobi identity. That is we only need to consider $\SN{f,\SN{f,f}_{J}}_{J}=0$ for an arbitrary even function. Thus, via direct computation we have
    \begin{eqnarray}
   \nonumber  \SN{f,\SN{f,f}_{J}}_{J} &=& (-1)^{\widetilde{C}}\left( S^{CD} \frac{\partial S^{BA}}{\partial x^{D}} + Q^{C} S^{BA}  \right)\frac{\partial f}{\partial x^{A}}\frac{\partial f}{\partial x^{B}}\frac{\partial f}{\partial x^{C}}\\
    \nonumber &-& f \left( 2 (-1)^{\widetilde{B}} S^{BC}\frac{\partial Q^{A}}{\partial x^{C}} + Q^{C} \frac{\partial S^{BA}}{\partial x^{C}}  \right)\frac{\partial f}{\partial x^{A}}\frac{\partial f}{\partial x^{B}}\\
    \nonumber &+& 2 f^{2}Q^{B}\frac{\partial Q^{A}}{\partial x^{B}}\frac{\partial f}{\partial x^{A}},
    \end{eqnarray}
collecting terms order by order in ``$\frac{\partial f}{\partial x}$".   Note that all terms involving higher order derivatives exactly cancel. Then we see that $\SN{f,\SN{f,f}_{J}}_{J}=0$ given the conditions on $S$ and $Q$ to form an odd Jacobi structure on $M$. Thus, the Jacobi identity is satisfied.
\item Via direct computation in local coordinates it is easy to see that
\begin{equation}
\nonumber \SN{f,gh}_{J} = \SN{f,g}_{J} h + (-1)^{(\widetilde{f}+1)\widetilde{g}} g \SN{f,h}_{J} - (-1)^{\widetilde{f}} Q^{A} \frac{\partial f}{\partial x^{A}}   gh.
\end{equation}
Thus the ``anomaly" is given by  $\SN{f, \mathbbmss{1}}_{J}$.
\end{enumerate}
\end{proof}

\begin{definition}
Given a function $f \in C^{\infty}(M)$ the associated \textbf{Hamiltonian vector field} is given by
\begin{eqnarray}
f &\rightsquigarrow & X_{f} \in \Vect(M)\\
\nonumber X_{f}(g) &=& (-1)^{\widetilde{f}} \SN{f,g}_{J} - Q(f) g.
\end{eqnarray}
\end{definition}

In natural local coordinates the Hamiltonian vector field of a function $f$ is

\begin{equation}
X_{f} = (-1)^{\widetilde{A}\widetilde{f}+1} S^{AB}\frac{\partial f}{\partial x^{B}} \frac{\partial }{\partial x^{A}} + (-1)^{\widetilde{f}}f Q^{A}\frac{\partial}{\partial x^{A}}.
\end{equation}

We will explore the properties of Hamiltonian vector fields later. Before we do this, let us examine morphisms of odd Jacobi manifolds.\\

\begin{definition}
Let $(M_{1}, S_{1}, Q_{1})$ and $(M_{2}, S_{2}, Q_{2})$ be odd Jacobi manifolds. Then a smooth map
\begin{equation}
\phi: M_{1} \rightarrow M_{2},
\end{equation}
is said to be an \textbf{odd Jacobi morphism} if and only if
\begin{equation}
\phi^{*}\SN{f,g}_{J_{2}} = \SN{\phi^{*}f , \phi^{*}g}_{J_{1}},
\end{equation}
for all $f,g \in C^{\infty}(M_{2})$.
\end{definition}

In other words, a morphism $\phi: M_{1} \rightarrow M_{2}$ is an odd Jacobi morphism if the associated pull-back morphism  is a homomorphism of odd Lie algebras.  Odd Jacobi manifolds form a category under composition of odd Jacobi morphisms.\\

\begin{proposition}
Let $(M, S, Q)$ be an odd Jacobi manifold and   $\phi : M \rightarrow M$ a diffeomorphism. Then the following are equivalent:
\begin{enumerate}
\item $\phi$ is an odd Jacobi (auto)morphism;  $\phi^{*} \SN{f,g}_{J} = \SN{\phi^{*}f, \phi^{*}g}_{J} $.
\item $\phi^{*}S = S$  and $\phi^{*}\Q = \Q$.
\item $X_{f}$ and $X_{\phi^{*}f}$ are $\phi$-related.
\end{enumerate}
\end{proposition}

\begin{proof}
$1. \Longleftrightarrow 2.$ follows from the fact that pull-back associated with the diffeomorphism $\phi$ is a symplectomorphism on the cotangent bundle and the definition of the odd  Jacobi bracket. Explicitly:
\begin{eqnarray}\nonumber
\phi^{*}\SN{f,g}_{J} &=& \phi^{*}\left( (-1)^{\widetilde{f}+1} \{ \{ S,f \},g \} - (-1)^{\widetilde{f}+1} \{\Q, fg  \} \right)\\
\nonumber &=& (-1)^{\widetilde{f}+1} \{ \{ \phi^{*}S,\phi^{*}f \},\phi^{*}g \} - (-1)^{\widetilde{f}+1} \{\phi^{*}\Q, \phi^{*}(fg)  \}.
\end{eqnarray}
Then via 1. we obtain 2.\\

$2.  \Longleftrightarrow 3.$ follows similarly. Explicitly:

\begin{eqnarray}\nonumber
\phi^{*}(X_{f}(g)) &=& \phi^{*}\left((-1)^{f} \SN{f,g}_{J} - Q(fg)   \right)\\
\nonumber &=& (-1)^{f} \SN{\phi^{*}f,\phi^{*}g}_{J} - \{ \Q, \phi^{*}(fg)  \}
\end{eqnarray}
where we have used  2. Thus,

\begin{equation}\nonumber
 \phi^{*}(X_{f}(g)) = X_{\phi^{*}f}(\phi^{*}g),
 \end{equation}
 Thus $X_{f}$ and $X_{\phi^{*}f}$ are $\phi$-related.

\end{proof}

Taking the nomenclature from classical mechanics, we will say that an odd Jacobi automorphism is a \emph{canonical transformation} with respect to the odd Jacobi bracket.

\begin{definition}
A vector field $X \in \Vect(M)$ is said to be a \textbf{Jacobi vector field} if and only if
\begin{equation}
L_{X}S =  \{ \chi, S \} = 0 \hspace{35pt} \textnormal{and}  \hspace{15pt} L_{X}Q = \{\chi, \Q  \} =  0,
\end{equation}
where $\chi \in C^{\infty}(T^{*}M)$ is the symbol or ``Hamiltonian" of the vector field $X$.
\end{definition}

The above definition is the infinitesimal version of a Jacobi automorphism. Note that the homological vector field $Q$ is a Jacobi vector field. \\

\begin{lemma}\label{lemma:jacobivectorfields}
Let $X \in \Vect(M)$ be a vector field on an odd Jacobi manifold. Then the following are equivalent:
\begin{enumerate}
\item $X$ is a Jacobi vector field.
\item $X$ is a derivation over the odd Jacobi bracket;
\begin{equation}\nonumber
X(\SN{f,g}_{J}) = \SN{X(f), g}_{J} + (-1)^{\widetilde{X}(\widetilde{f}+1)} \SN{f,X(g)}_{J}.
\end{equation}
\item $[X, Y_{f}] = (-1)^{\widetilde{X}} Y_{X(f)}$, for all Hamiltonian vector fields $Y_{f}$.
\end{enumerate}
\end{lemma}

\begin{proof}
Let $X \in \Vect(M)$ be a vector field. Then consider the  symbol or ``Hamiltonian" of such a vector field: $X^{A} \frac{\partial}{\partial x^A}  \rightarrow \chi = X^{A}p_{A} \in C^{\infty}(T^{*}M)$. The Lie derivative with respect to the vector field acting on $C^{\infty}(T^{*}M)$ is just $L_{X} = \{\chi, \bullet \}$.\\

1. $\Longleftrightarrow$ 2.  is proved via successive use of the Jacobi identity for the canonical Poisson bracket together with 1. Explicitly:
\begin{eqnarray}
\nonumber X(\SN{f,g}_{J}) &=&  (-1)^{(\widetilde{f} + \widetilde{X})+1} \{ \{S, \{ \chi,f \}  \},g \} - (-1)^{(\widetilde{f} + \widetilde{X})+1} \{ \Q , \{\chi,f \}g\}\\
\nonumber &+& (-1)^{\widetilde{X}(\widetilde{f}+1) + \widetilde{f}+1} \{ \{ S,f\}, \{\chi, g  \}  \} - (-1)^{\widetilde{X}(\widetilde{f}+1) + \widetilde{f}+1} \{\Q,f \{\chi,g \} \},
\end{eqnarray}
which establishes the result.\\

1. $\Longleftrightarrow$ 3. via direct computation. Explicitly:
\begin{eqnarray}\nonumber
X(Y_{f}(g)) &=& (-1)^{\widetilde{f}} \SN{X(f),g}_{J} + (-1)^{\widetilde{f}+ \widetilde{X}(\widetilde{f}+1)} \SN{f,X(g)}_{J}\\
\nonumber &-& X(Q(f))g - (-1)^{ \widetilde{X}(\widetilde{f}+1)} Q(f)Q(g),
 \end{eqnarray}
 using the derivation property of $X$ over the odd Jacobi bracket. Then
 \begin{equation}\nonumber
 Y_{f}(X(g)) = (-1)^{\widetilde{f}}\SN{f,X(g) }_{J} - Q(f)X(g),
 \end{equation}
 gives
 \begin{equation}\nonumber
 [X,Y_{f}](g) = (-1)^{\widetilde{f}}\SN{X(f), g} - (-1)^{\widetilde{X}} Q(X(f))g,
 \end{equation}
 using 1. Thus the result is established.

\end{proof}

\begin{corollary}\label{corol:homologicalvectorfield}
The homological vector field $Q$ satisfies the following:
\begin{enumerate}
\item $Q(\SN{f,g}_{J}) = \SN{Q(f), g} + (-1)^{\widetilde{f}+1} \SN{f, Q(g)}$,
\item $[Q, X_{f}] = - X_{Q(f)}$,
\end{enumerate}
for all $f,g \in C^{\infty}(M)$.
\end{corollary}

\begin{proposition}\label{prop:morphism}
The assignment $f \rightsquigarrow X_{f}$ is a morphism  between the odd Lie algebra on $C^{\infty}(M)$ provided by the odd Jacobi brackets and the Lie algebra of vector fields. Specifically, the following holds:
\begin{equation}
[X_{f}, X_{g}] = - X_{\SN{f,g}_{J}}
\end{equation}
for all $f,g \in C^{\infty}(M)$.
\end{proposition}

\begin{proof}
Writing out the commutator explicitly we obtain
\begin{eqnarray}
\nonumber [X_{f}, X_{g}](h)&=& (-1)^{\widetilde{f} + \widetilde{g}}\SN{f,\SN{g,h}}_{J}- (-1)^{\widetilde{f} + \widetilde{g} + (\widetilde{f}+1)(\widetilde{g}+1)}\SN{g,\SN{f,h}}_{J}\\
\nonumber &-& (-1)^{\widetilde{f}} \SN{f, Q(g)}_{J} h + (-1)^{\widetilde{g} +(\widetilde{f}+1)(\widetilde{g}+1)} \SN{g, Q(f)}_{J} h\\
\nonumber & + & \textnormal{ terms that cancel}.
\end{eqnarray}
In the above we have explicitly used the generalised Leibniz rule for the odd Jacobi bracket. Importantly all other possible terms cancel. Then using the Jacobi identity for the odd Jacobi brackets in the form
\begin{equation}\nonumber
\SN{\SN{f,g}_{J},h}_{J} = \SN{f, \SN{g,h}_{J}}_{J} - (-1)^{(\widetilde{f}+1)(\widetilde{g}+1)} \SN{g, \SN{f,h}_{J}}_{J},
\end{equation}

and the differential property of $Q$ over the odd Jacobi bracket we obtain

\begin{equation}\nonumber
[X_{f}, X_{g}] = (-1)^{\widetilde{f} + \widetilde{g}} \SN{\SN{f,g}_{J}, h}_{J} + \left(Q(\SN{f,g}_{J})\right)h,
\end{equation}
and thus the result is established.
\end{proof}

Unlike the Schouten or indeed the Poisson case, Hamiltonian vector fields on a Jacobi manifold (both even or odd) in general do not generate infinitesimal automorphisms of the bracket structure.

\begin{proposition}\label{prop:Hamiltonianvectorfields}
A Hamiltonian vector field $X_{f} \in \Vect(M)$ is a Jacobi vector field if and only if $f \in C^{\infty}(M)$ is Q-closed.
\end{proposition}

\begin{proof}
If  and only if $f$ is Q-closed, i.e. $Q(f) =0$ then $X_{f} = (-1)^{\widetilde{f}}\SN{f, \bullet}_{J}$. Using Proposition \ref{prop:morphism}. we have
\begin{equation}\nonumber
[X_{f}, X_{g}] = - X_{\SN{f,g}_{J}} = (-1)^{\widetilde{f}+1} X_{X_{f}(g)},
\end{equation}
for any $g \in C^{\infty}(M)$. Then via Lemma \ref{lemma:jacobivectorfields}. the result is established.
\end{proof}

\subsection{Basic examples}\label{subsec:examples}

Let us turn our attention to straight forward examples to show that the category of odd Jacobi manifolds is not completely empty.

\begin{example}\label{R11}
\emph{The superline}\\
Consider the supermanifold $\mathbbmss{R}^{1|1}$, which we equip with local coordinates $(t, \xi)$. Here $t$ is commuting and $\xi$ anticommuting. There is a canonical odd Jacobi structure  on $\mathbbmss{R}^{1|1}$ given by
\begin{eqnarray}
\nonumber S &=& - \pi p,\\
\nonumber \Q &=& - \pi,
\end{eqnarray}
where we have used fiber coordinates $(p, \pi)$ on $T^{*}(\mathbbmss{R}^{1|1})$. Direct computation produces
\begin{equation}\nonumber
\{S,S  \}_{T^{*}(\mathbbmss{R}^{1|1})} = -2 \: (-\pi) (-\pi p),
\end{equation}
which established the fact we have an odd Jacobi structure. Note that as $\pi^{2}=0$, $S$ can also be considered as a Schouten structure. The odd Jacobi brackets are given by
\begin{eqnarray}
\nonumber \SN{f,g}_{J} &=& (-1)^{\widetilde{f}}\frac{\partial f}{\partial \xi}\frac{\partial g}{\partial t}- \frac{\partial f}{\partial t}\frac{\partial g}{\partial \xi}\\
\nonumber &-& (-1)^{\widetilde{f}}\left( \frac{\partial f}{\partial \xi} \right)g - f \left( \frac{\partial g}{\partial \xi} \right),
\end{eqnarray}
which we recognise as the canonical Schouten bracket on $\mathbbmss{R}^{1|1}$ plus a term that spoils the strict Leibniz rule.
\end{example}

\begin{example}
\emph{Schouten manifolds}\\
These are understood as odd Jacobi manifolds for which $Q=0$. Schouten manifolds are of particular interest in mathematical physics due to their connection with the BV-antifield formalism.
\begin{enumerate}
\item Lie--Schouten structures: A vector space $\mathfrak{g}$ is a Lie algebra if and only if $\Pi \mathfrak{g}$ comes equipped with a weight minus one (``linear") homological vector field. Here $\Pi$ is the parity reversion functor.  If we employ local coordinates $(\xi^{\alpha})$ on $\Pi \mathfrak{g}$ then the homological vector field is given by:
    \begin{equation}\nonumber
    Q_{\mathfrak{g}} = \frac{1}{2} \xi^{\alpha}\xi^{\beta}Q^{\gamma}_{\beta \alpha} \frac{\partial}{\partial \xi^{\gamma}}.
     \end{equation}
     A weight minus one Schouten structure on $T^{*}(\Pi \mathfrak{g}^{*})$ can be associated with the homological vector field. Employing natural local coordinates $(\eta_{\alpha}, \pi^{\alpha})$  the Schouten structure is given by
     \begin{equation}\nonumber
     S = \frac{1}{2}(-1)^{\widetilde{\alpha} + \widetilde{\beta}} \pi^{\alpha} \pi^{\beta}Q^{\gamma}_{\alpha \beta} \eta_{\gamma}.
     \end{equation}
     The Schouten condition $\{ S,S \}=0$, is equivalent to the homological condition on $Q_{\mathfrak{g}}$ and in turn is equivalent to the Jacobi identity on the initial Lie algebra structure $\mathfrak{g}$.
\item Odd symplectic manifolds: An odd symplectic manifold is defined as a manifold $M$ of dimensions $(n|n)$ equipped with a closed non-degenerate odd two form denoted $\omega$, understood as a function on the total space of $\Pi TM$ . In natural local coordinates the odd symplectic form is given by
    \begin{equation}\nonumber
    \omega = \frac{1}{2} dx^{A} dx^{B}\omega_{BA}(x).
    \end{equation}
    The non degeneracy condition means that, as an matrix $\omega_{BA}$ is invertible. Let us denote this inverse by $\omega^{AB}$. Then associated with an odd symplectic form is a Schouten structure given by
    \begin{equation}\nonumber
    S = \frac{1}{2}(-1)^{\widetilde{B}} \omega^{AB}p_{B}p_{A} \in C^{\infty}(T^{*}M).
    \end{equation}
    The Schouten condition $\{S,S  \} =0$ is directly equivalent to the closed of $\omega$.
\end{enumerate}
\end{example}

\begin{example}
\emph{Q-manifolds}\\
These are understood as odd Jacobi manifolds with $S=0$. The  odd Jacobi bracket  on a Q-manifold is given by $\SN{f,g}_{Q} = (-1)^{\widetilde{f}} Q(fg)$. Such manifold, and in particular their algebra of functions is of wide interest in mathematics due to the fact that many algebraic structures can be encoded on formal Q-manifolds.
\begin{enumerate}
\item The de Rham complex: The algebra of differential (pseudo)forms over a manifold $M$ is understood as the algebra of  functions on the manifold $\Pi TM$. The de Rham differential is the canonical homological vector field $d = dx^{A} \frac{\partial}{\partial x^{A}}$, where we have employed natural local coordinates $(x^{A}, dx^{A})$. Then over any manifold, the algebra of differential forms comes equipped with an odd Jacobi bracket given by $\SN{\alpha, \beta}_{d} = (-1)^{\widetilde{\alpha}} (d\alpha) \beta + \alpha (d\beta)$ with $\alpha, \beta \in C^{\infty}(\Pi TM)$.
\item Lie algebras: Let the vector space $\mathfrak{g}$ be a Lie algebra. Then as seen in the previous example, $(\Pi \mathfrak{g}, Q_{\mathfrak{g}})$ is a Q-manifold. Thus, $C^{\infty}(\Pi \mathfrak{g}) $ has an odd Jacobi bracket (of weight one in linear coordinates) given by
    \begin{equation}\nonumber
    \SN{f,g}_{Q} = \left( (-1)^{\widetilde{f}} \frac{1}{2}\xi^{\alpha} \xi^{\beta} Q_{\beta \alpha}^{\gamma} \frac{\partial f}{\partial \xi^{\gamma}}\right)g + f \left(\frac{1}{2}\xi^{\alpha} \xi^{\beta} Q_{\beta \alpha}^{\gamma} \frac{\partial g}{\partial \xi^{\gamma}}\right).
    \end{equation}
 \end{enumerate}
\end{example}

\begin{remark}
It is worth noting that this construction of an odd Jacobi bracket generalises directly to $L_{\infty}$-algebras (c.f. \cite{Lada:1992wc}), which can be understood in terms of homological vector fields inhomogeneous in the linear coordinate $\xi$. Thus, the associated odd Jacobi bracket is also inhomogeneous in weight, yet remarkably binary. Similarly, one can build an odd Jacobi bracket associated with a Lie algebroid or even an $L_{\infty}$-algebroid.
\end{remark}

\begin{example}\label{example:oddcontact}
\emph{Odd contact manifolds}\\
Rather than diverge into the general theory of contact manifolds let us explore a specific example. Consider the manifold $M := \Pi T^{*}N\times \mathbb{R}^{0|1}$, where $N$ is a pure even (classical) manifold\footnote{Due to a Darboux theorem for contact manifolds this example is also generic. One can generalise the arguments of Arnold \cite{Arnold1989} (see Appendix 4) without much difficulty to include odd contact structures on supermanifolds.}. \\

Let us employ natural local coordinates $(x^{a}, x^{*}_{a}, \tau)$. The coordinates $x^{a}$ are even, while the other coordinates $x^{*}_{a}$ and $\tau$ are odd. The dimension of $M$ is $(n|n+1)$ assuming the dimension of $N$ is $n$.\\

The manifold $\Pi T^{*}N$ comes equipped with a canonical odd symplectic structure  $\omega = - dx^{*}_{a} dx^{a}$.  Also note that all odd sympelctic manifolds are equivalent  to an anticotangent bundle and that the base manifold can be chosen to be a classical manifold \cite{Khudaverdian2004,Schwarz1993}.  The manifold $M$ comes equipped with an \emph{odd contact one form}, which is the even one form
\begin{equation}\nonumber
\alpha = d \tau - x^{*}_{a}dx^{a}.
\end{equation}
We will for the purposes of this example stipulate that  the above is the correct form via ``superisation" of the even contact structure on $\mathbb{R}^{2n +1}$. Clearly the two form  $d \alpha = - dx^{*}_{a}dx^{a}$ is, unsurprisingly,  the canonical odd symplectic structure on $\Pi T^{*}N \subset M $.  We employ natural fibre coordinates $(dx^{a}, dx^{*}_{a}, d\tau)$ on $\Pi T M$ and  $(p_{a}, p^{a}_{*}, \pi)$ on $T^{*}M$. The coordinates $dx^{*}_{a}$, $d\tau$, $p_{a}$ are even and $dx^{a}$, $p^{a}_{*}$, $\pi$ are odd.  Then let us define the almost Schouten structure via

\begin{equation}\nonumber
\phi_{S}^{*}(\alpha) =0, \hspace{15pt}  \textnormal{and} \hspace{15pt} \phi^{*}_{S}(d \alpha) = S,
\end{equation}

where we have the standard fibre-wise morphism $\phi_{S}: T^{*}M \rightarrow \Pi T M$ associated with  any almost Schouten structure.  Let us take the \emph{Ansatz}
\begin{equation}\nonumber
S = p^{a}_{*}  \left( p_{a} + x^{*}_{a} \pi\right) \in C^{\infty}(T^{*}M),
\end{equation}

based on standard constructions related to even contact structures. Then

\begin{eqnarray}
\nonumber \phi^{*}_{S}(d x^{a}) &=& \frac{\partial S}{\partial p_{a}} =  p_{*}^{a}.\\
\nonumber \phi^{*}_{S}(d x^{*}_{a}) &=& - \frac{\partial S}{\partial p^{a}_{*}} = -\left(p_{a} +  x^{*}_{a} \pi  \right).\\
\nonumber \phi^{*}_{S}(d \tau) &=& -\frac{\partial S}{\partial \pi} =  -p_{*}^{a} x^{*}_{a}.
\end{eqnarray}

Direct computation gives

\begin{equation}\nonumber
\phi_{S}^{*}(\alpha) =  - p^{a}_{*}x^{*}_{a} - x^{*}_{a} p_{*}^{a} =0, \hspace{15pt}\textnormal{and} \hspace{15pt}\phi_{S}^{*}(d\alpha) = (p_{a} + x^{*}_{a}\pi)p_{*}^{a} = S.
\end{equation}

Thus our \emph{Ansatz} is confirmed to be correct. To extract the homological vector field one  needs to calculate the self Poisson bracket of the almost Schouten structure. Explicitly

\begin{eqnarray}
\nonumber \{ S,S \}_{T^{*}M} &=& 2 \frac{\partial S}{\partial p_{*}^{a}}\frac{\partial S}{\partial x^{*}_{a}}\\
\nonumber &=& 2 \left((p_{a} + x^{*}_{a} \pi)(-1)p^{a}_{*}\pi  \right)\\
 \nonumber &=& - 2 (-\pi) \left(p^{a}_{*}(p_{a} + x^{*}_{a}\pi)  \right),
\end{eqnarray}

 thus  $\Q = - \pi$. It is easy to see that $\{ \Q, S \}_{T^{*}M} = 0$ and $\{ \Q, \Q \}_{T^{*}M} =0$. Furthermore, notice that the homological vector field $Q = - \frac{\partial}{\partial \tau}$ satisfies

\begin{equation}\nonumber
i_{Q}\alpha = 1 \hspace{15pt} \textnormal{and} \hspace{15pt} i_{Q}(d\alpha) = 0.
\end{equation}

\noindent \textbf{Statement:} adding an extra odd variable (``odd time") to the canonical odd symplectic manifold $\Pi T^{*}N$ produces a manifold with an odd Jacobi bracket rather than a Schouten bracket.\\

 In natural local coordinates the odd Jacobi bracket is given by

\begin{eqnarray}
\nonumber \SN{f,g}_{J} &=& (-1)^{\widetilde{f}+ 1} \frac{\partial f}{\partial x^{*}_{a}} \frac{\partial g}{\partial x^{a}} - \frac{\partial f}{\partial x^{a}} \frac{\partial g}{\partial x^{*}_{a}}\\
\nonumber &+& x^{*}_{a} \frac{\partial f}{\partial x^{*}_{a}} \frac{\partial g}{\partial \tau} - (-1)^{\widetilde{f} + 1} \frac{\partial f}{\partial \tau} x^{*}_{a}\frac{\partial g}{\partial x^{*}_{a}}\\
\nonumber &+& f \frac{\partial g}{\partial \tau} - (-1)^{\widetilde{f}+1} \frac{\partial f}{\partial \tau} g.
\end{eqnarray}
We will return to this example in Section \ref{sec:odd contact}.

\end{example}

\subsection{The Schoutenization of an odd Jacobi structure}\label{subsec:Schoutenization}
In this section we show that via an extension of an odd Jacobi manifolds one can canonically construct a QS-manifold. This mimics very closely the classical situation of Poissonization. \\

Consider the manifold $M \times \mathbbmss{R}$, which we equip with natural coordinates $(x^{A}, t)$. Here $t$ is the commuting coordinate on the factor $\mathbbmss{R}$. Assuming that $M$ is in fact an odd Jacobi manifold one can build

\begin{equation}
\bar{S} = e^{-t}\left( S- \Q p \right) \in C^{\infty}( T^{*}(M \times \mathbbmss{R})),
\end{equation}

where $p$ is the momenta associated with $t$.

\begin{theorem}\label{th:schoutenization}
Let $(M, S, \Q)$ be an odd Jacobi manifold. Then the triple $(M\times \mathbbmss{R}, \bar{S}, \Q)$ is a QS-manifold in the sense of Voronov.
\end{theorem}

\begin{proof}
First observe that $\bar{S}$ is of order two in the momenta and is Grassmann odd. Second note the Poisson bracket on $T^{*}(M \times \mathbbmss{R})$ has the natural decomposition as $\{ ,\} = \{, \}_{T^{*}M} + \{ , \}_{T^{*}\mathbbmss{R}}$. It is then a straight forward exercise to take into account terms that contain conjugate variables and those that do not to show that
\begin{enumerate}
\item $\{ \bar{S}, \bar{S} \} = e^{-2t}\left( \{S,S \}+ 2 \Q S - 2p \{S, \Q\} +p^{2}\{\Q,\Q \}\right)$,
\item$\{\bar{S}, \Q \} = e^{-t} \{ S, \Q \} - p\{\Q, \Q \} $,
\end{enumerate}
and thus as $(M, S, \Q)$ be an odd Jacobi manifold we establish that $\{\bar{S}, \bar{S}\}=0$ and $\{ \bar{S}, \Q\}=0$. This establishes the proposition.
\end{proof}

\noindent \textbf{Statement:} adding  ``time" to an odd Jacobi manifold produces a QS-manifold.

\subsection{Exact QS-manifolds and odd Jacobi structures}\label{subsec:exactQS}

In this section we define the notion of an exact QS-manifold, understood as a fairly direct generalisation of an exact or homogeneous Poisson manifold. We then, taking Petalidou \cite{Petalidou2002} as our inspiration establish a link between exact QS-structures and odd Jacobi structures. In particular these structures are on the same manifold and no extension is required. This is in contrast to the process of Schoutenization of an odd Jacobi manifold, which requires the manifold to be extended by one even direction.

\begin{definition}
An \textbf{exact QS-manifold} is the quadruple $(M, \bar{S}, \bar{Q}, \bar{E})$, where $(M,  \bar{S}, \bar{Q})$ is a QS-manifold and $\bar{E} \in \Vect(M)$ is an even vector field, referred as the \textbf{homothety vector field} that  satisfies
\begin{equation}
L_{\bar{E}} \bar{S} = - \bar{S} \hspace{15pt} \textnormal{and} \hspace{15pt} L_{\bar{E}}\bar{Q} = -\bar{Q}.
\end{equation}
\end{definition}

The existence of the homothety vector field on a QS-manifold means that  both $\bar{\Q}$ and $\bar{S}$ are \emph{exact}

\begin{eqnarray}
 \nonumber \{ \mathcal{E}, \bar{S} \} = -\bar{S} &\longrightarrow& \bar{S} = \{\bar{S}, \mathcal{E}  \},\\
\nonumber \{ \mathcal{E}, \bar{\Q} \} = -\bar{\Q} &\longrightarrow& \bar{\Q} = \{ \bar{\Q}, \mathcal{E} \},
\end{eqnarray}

with respect to the  operators on $C^{\infty}(T^{*}M)$ they generate. Here $\mathcal{E} \in C^{\infty}(T^{*}M)$ is the symbol of the homothety vector field. In other words, the Schouten structure is itself a trivial element in the \emph{Schouten cohomology}  as generated by $\delta_{\bar{S}} := \{\bar{S}, \bullet \}$. The homological structure is similarly  a trivial element in the  cohomology of the operator $L_{\bar{Q}}$. Poisson cohomology goes back to Lichnerowicz \cite{Lichnerowicz1977}, who also introduced the notion of (even) Jacobi manifolds.  For a discussion of the cohomology of a Q-manifold see \cite{Lyakhovich2004,Lyakhovich2010}.\\

Let us now proceed to the theorem relating exact QS structures  to odd Jacobi structures on the \emph{same underlying manifold}.

\begin{theorem}\label{theorem:exactQS}
Let $(M, \bar{S},\bar{Q},  \bar{E})$ be an exact QS-manifold. Then the  pair $(S = \bar{S} + \mathcal{E}\bar{\Q},  Q = \bar{Q})$,
provides an odd Jacobi structure on the manifold $M$.
\end{theorem}

\begin{proof}
The proof requires one to examine the the invariance and compatibility conditions for odd Jacobi structures. The homological condition is given.
\begin{itemize}
\item Writing out the self-Poisson bracket of $S$ one obtains
 \begin{eqnarray}
\nonumber \{S, S \} &=& \{\bar{S}, \bar{S} \} + 2 \{ \bar{S}, \mathcal{E} \}\bar{\Q} + 2 \mathcal{E}\{ \bar{S}, \bar{\Q} \}\\
\nonumber &-& \{ \mathcal{E}, \mathcal{E} \}\bar{\Q}^{2} - 2 \mathcal{E}\{ \mathcal{E}, \bar{\Q} \}\bar{\Q} + \mathcal{E}^{2}\{\bar{\Q}, \bar{\Q} \}\\
\nonumber &=& 2 \bar{\Q}\left(\{ \mathcal{E} , \bar{S}\}+ \mathcal{E} \{\mathcal{E}, \bar{\Q} \}\right)\\
\nonumber &=& - 2 \bar{\Q}\left ( \bar{S}  + \mathcal{E} \bar{\Q}\right).
\end{eqnarray}
\item Writing out the Poisson bracket between $\mathcal{Q}$ and $S$ one obtains
\begin{eqnarray}
\nonumber \{\mathcal{Q} , S \} &=& \{ \bar{\Q} , \bar{S}\} + \{\bar{\Q}, \mathcal{E}  \}\bar{\Q} + \mathcal{E}\{ \bar{\Q}, \bar{\Q}\}\\
\nonumber &=&  \bar{\Q}^{2} =0.
\end{eqnarray}
\end{itemize}
Thus $S$ and $Q$ define an odd Jacobi structure on the manifold $M$.
\end{proof}

Via a mild generalisation of the above proof we arrive at the following corollary:

\begin{corollary}
  Associated with any exact QS-structure on $M$ is a pencil of odd Jacobi structures also on $M$ given by  $(S = a \bar{S} + b \mathcal{E}\bar{\Q}, Q =  b\bar{Q})$ where $a,b$ are  even parameters (or just real  numbers).
\end{corollary}

\noindent\textbf{Statement:} every exact QS-manifold is also an odd Jacobi manifold.\\

 Setting $a=b=1$ produces a ``canonical" odd Jacobi structure on $M$. Setting $a=1$ and $b=0$ confirms the notion that a Schouten manifold can be thought of as an odd Jacobi manifold with the trivial homological vector field. Setting $a=0$ and $b=1$ confirms the notion that a Q-manifolds can be thought of as an odd Jacobi manifold with the trivial Schouten structure. In a loose sense, intermediate values of $a$ and $b$ interpolate between the extremes of Schouten manifolds and Q-manifolds understood as examples of odd Jacobi manifolds.

\section{Jacobi Algebroids}

\subsection{Quasi Q-manifolds and Jacobi algebroids}\label{sec:main constructions}

In this section we propose a definition of a Jacobi algebroid in terms of an odd Jacobi structure on the total space of $\Pi E^{*}$, given a vector bundle $E \rightarrow M$. It will turn out that this definition is equivalent to that given by Grabowski \& Marmo \cite{Grabowski2001} (also see Iglesias \& Marrero \cite{Iglesias2001}). We postpone the details of this equivalence to the next section and take the following definition as the starting point of this work.

\begin{definition}
A vector bundle $E \rightarrow M$ is said to have the structure of a \textbf{Jacobi algebroid} if and only if the total space of $\Pi E^{*}$ comes equipped with a weight minus one odd Jacobi structure.
\end{definition}

Recall that an odd Jacobi structure on a manifold is a pair of odd functions on the total space of the cotangent bundle quadratic and linear in the fibre coordinates together with a series of conditions expressed in terms of the canonical Poisson bracket. Let us employ natural local coordinates $(x^{A}, \eta_{\alpha}, p_{A}, \pi^{\alpha})$ on the total space of $T^{*}(\Pi E^{*})$. The weight is assigned  as $\w(x^{A}) = 0$, $\w(p_{A})=0$, $\w(\eta_{\alpha}) = +1$ and $\w(\pi^{\alpha}) = -1$.  This is the \emph{natural weight} associated with the vector bundle structure $E^{*} \rightarrow M$. The parity of the coordinates is given by $ \widetilde{x}^{A}=  \widetilde{A}$, $\widetilde{\eta}_{\alpha}= (\widetilde{\alpha} +1)$, $\widetilde{p}_{A}= \widetilde{A}$ and  $\widetilde{\pi}^{\alpha} =  (\widetilde{\alpha}+1)$. In these natural local coordinates the odd Jacobi structure is given by

\begin{eqnarray}
S &=&(-1)^{\widetilde{\alpha}}\pi^{\alpha}Q_{\alpha}^{A}(x)p_{A}+ (-1)^{\widetilde{\alpha} + \widetilde{\beta}}\frac{1}{2}\pi^{\alpha}\pi^{\beta}Q_{\beta \alpha}^{\gamma}\eta_{\gamma},\\
\nonumber  \Q &=& \pi^{\alpha}Q_{\alpha}(x),
\end{eqnarray}

which are both functions on the total space of $T^{*}(\Pi E^{*})$. The notation and the sign factors employed  make clear the relation with Lie algebroids.   \\

This structure satisfies the conditions:
\begin{enumerate}
\item $\{\Q,\Q  \}_{T^{*}(\Pi E^{*})} =0$. \label{homcond}
\item $\{ \Q, S \}_{T^{*}(\Pi E^{*})} = 0$.
\item $\{S, S  \}_{T^{*}(\Pi E^{*})} =  - 2 \Q S$.
\end{enumerate}

Setting $\Q =0$ means that $S$ is a Schouten structure and thus we have a genuine Lie algebroid.  Note that due to the fact that the function $\Q$ does not contain conjugate variables the condition \ref{homcond}. is automatically satisfied.  This is not generally the case and typically \ref{homcond}. will be a non-trivial condition.\\

The Jacobi algebroid structure on the vector bundle $E \rightarrow M$ is directly equivalent to the existence of a weight minus one odd Jacobi bracket on $C^{\infty}(\Pi E^{*})$. That is the algebra of  ``multivector fields" comes equipped with the structure of an odd Jacobi algebra \emph{viz}

\begin{equation}
\nonumber \SN{X,Y}_{E}  = (-1)^{\widetilde{X}+1}\{  \{S, X   \}_{T^{*}(\Pi E^{*})} , Y\}_{T^{*}(\Pi E^{*})}- (-1)^{\widetilde{X}+1} \{ \Q, XY \}_{T^{*}(\Pi E^{*})},
\end{equation}

with $X,Y \in C^{\infty}(\Pi E^{*})$.\\

In natural local coordinates this bracket is given by

\begin{eqnarray}\nonumber
\SN{X,Y}_{E} &=& Q_{\alpha}^{A}\left((-1)^{(\widetilde{X}+ \widetilde{\alpha}+1)(\widetilde{A}+1) } \frac{\partial X}{\partial \eta_{\alpha}}\frac{\partial Y}{\partial x^{A}}  - (-1)^{(\widetilde{X}+1)\widetilde{\alpha}} \frac{\partial X}{\partial x^{A}}\frac{\partial Y}{\partial \eta_{\alpha}}\right)\\
\nonumber &-& (-1)^{(\widetilde{X}+1)\widetilde{\alpha}+ \widetilde{\beta}}Q_{\alpha \beta}^{\gamma}\eta_{\gamma}\frac{\partial X}{\partial\eta_{\beta}}\frac{\partial Y}{\partial \eta_{\alpha}}\\
\nonumber &+&(-1)^{\widetilde{X}}Q_{\alpha}\frac{\partial X}{\partial \eta_{\alpha}} Y  + X Q_{\alpha}\frac{\partial Y}{\partial \eta_{\alpha}}.
\end{eqnarray}

Where $X = X(x, \eta) =  X(x) + X^{\alpha}(x) \eta_{\alpha} + \frac{1}{2!}X^{\alpha \beta}(x) \eta_{\beta}\eta_{\alpha} + \cdots$ \emph{etc}. The above odd Jacobi bracket is the natural generalisation of the weight minus one Schouten bracket associated with a Lie algebroid, which itself is a generalisation of the Schouten--Nijenhuis bracket between multivector fields over a manifold.\\

\begin{theorem}
The existence of a Jacobi algebroid structure on the vector bundle  $E \rightarrow M$ is equivalent to $\Pi E$ being  a weight one quasi Q-manifold.
\end{theorem}

\begin{proof}
Recall that the canonical double vector bundle morphism

\begin{equation}\nonumber
T^{*}(\Pi E^{*}) \stackrel{R}{\longrightarrow} T^{*}(\Pi E),
\end{equation}

is a symplectomorphism between the respective canonical symplectic structures.  We place details of this morphism in an  appendix. Thus we can \emph{move} the odd Jacobi structure from $\Pi E^{*}$ to $\Pi E$. However the resulting structure over $\Pi E$ will not be a genuine odd Jacobi structure as the degree in momenta (fibre coordinates of the cotangents) is not conserved under the canonical double vector bundle morphism.\\

Let us employ natural local coordinates $(x^{A}, \xi^{\alpha}, p_{A}, \pi_{\alpha})$ on $T^{*}(\Pi E)$. The weight of the coordinates is assigned as $\w(\xi^{\alpha})=-1$ and $\w(\pi_{\alpha})= +1$. The parities are $\widetilde{\xi}^{\alpha} = \widetilde{\pi}_{\alpha}= (\widetilde{\alpha}+1)$.  Then the canonical double vector bundle morphism is given by

\begin{equation}
\nonumber R^{*}\left( \pi_{\alpha} \right) =  \eta_{\alpha}, \hspace{25pt}  R^{*}\left(\xi^{\alpha}  \right)  = (-1)^{\widetilde{\alpha}} \pi^{\alpha}.
\end{equation}

Then let us consider

\begin{eqnarray}\nonumber
\hat{S} := (R^{-1})^{*}S  &=& \xi^{\alpha}Q_{\alpha}^{A}(x) p_{A}+ \frac{1}{2}\xi^{\alpha}\xi^{\beta}Q_{\beta \alpha}^{\gamma}(x)\eta_{\gamma},\\
\nonumber \hat{\Q} := (R^{-1})^{*}\Q &=& (-1)^{\widetilde{\alpha}}\xi^{\alpha}Q_{\alpha}(x),
\end{eqnarray}

both of which are functions on the total space of $T^{*}(\Pi E)$. Note that the function $\hat{S}$ is now linear in moneta and that $\hat{\Q}$ is independent of momenta.  As $R$ is a symplectomorphism we naturally have

\begin{equation}\nonumber
\{ \hat{S}, \hat{S} \}_{T^{*}(\Pi E)} =  - 2 \hat{\Q}\hat{S},  \hspace{15pt}\textnormal{and}\hspace{15pt} \{\hat{\Q}, \hat{S} \}_{T^{*}(\Pi E)}=0.
\end{equation}

Then we can ``undo" the symbol map which gives an odd vector field on $\Pi E$ and an odd function linear in the fibre coordinate:

\begin{eqnarray}\nonumber
\hat{S}&\longrightarrow& \D = \xi^{\alpha}Q_{\alpha}^{A}(x) \frac{\partial}{\partial x^{A}}+ \frac{1}{2}\xi^{\alpha}\xi^{\beta}Q_{\beta \alpha}^{\gamma}(x) \frac{\partial}{\partial \xi^{\gamma}} \in \Vect(\Pi E),\\
\nonumber \hat{\Q} &\longrightarrow& q =  - (-1)^{\widetilde{\alpha}}\xi^{\alpha}Q_{\alpha}(x)  \in C^{\infty}(\Pi E).
\end{eqnarray}

Note the extra minus sign in the definition of $q$. As the symbol map takes commutators of vector fields to Poisson brackets etc., it is not hard to see that the conditions that $(S, \Q)$  be an odd Jacobi structure translates to  $\Pi E$ being  a quasi Q-manifold: \\
\begin{center}
\begin{tabular}{lcl}
 $[\D,\D] =  2 q\D$,  & and  & $\D[q] =0$.
\end{tabular}
\end{center}
The grading is with respect to the natural grading associated with the vector bundle structure $E \rightarrow M$. That is we assign the weight as $\bar{\w}(x^{A}) =0$ and $\bar{\w}(\xi^{\alpha})= 1$. Note that $\bar{\w} = - \w$.
\end{proof}

\subsection{Lie algebroids in the presence of  a 1-cocycle}\label{sec:from jacobi algebroids}

In this section we  in essence restate  Grabowski \& Marmo's theorem 5 of \cite{Grabowski2001} giving a one-to-one correspondence between Jacobi algebroids and Lie algebroids in the presence of a 1-cocycle.\\

As we are considering the total space $\Pi E$ to be a graded manifold we naturally have an Euler vector field, which counts the weight of objects via it's Lie derivative. In natural local coordinates the Euler vector field is given by $\Xi = \xi^{\alpha}\frac{\partial}{\partial \xi^{\alpha}}$ as we have assigned weight $\bar{\w}(x) =0$ and $\bar{\w}(\xi)= 1$. A ``differential form"  $\omega \in C^{\infty}(\Pi E)$ is homogeneous and of  weight $p$ if $\Xi(\omega)= p \omega$. In a similar way, a vector field, $V \in \Vect(\Pi E)$ is homogeneous and of weight $r$ if $[\Xi,V] = r V$. The action of the Euler vector field can be extended to higher tensor objects, but we will have no call to use it in this work. In relation to Jacobi algebroids, we will be exclusively  interested in object of weight one. Such objects are invariant under the action of the Euler vector field, or in more classical language they are \emph{linear} objects.

\begin{proposition}\label{prop:quasiQalgebroid}
Let $(\Pi E , \D, q)$ be the weight one quasi Q-manifold associated with a  Jacobi algebroid. Then
\begin{equation}\nonumber
Q :=  \D - q \Xi,
\end{equation}
defines a homological vector field on $\Pi E$ of weight one and thus a Lie algebroid structure on $E \rightarrow M$. Furthermore, we have $Q(q)=0$ and thus we have a Lie algebroid in the presence of a 1-cocycle.
\end{proposition}

\begin{proof}
The weight conditions are clear from the definitions. We need to prove that $Q$ is homological. Explicitly
\begin{eqnarray}
\nonumber Q^{2}\omega &=& \D^{2}\omega + q \Xi \left(q \Xi(\omega) \right) - \D\left ( q\: \Xi(\omega)\right)- q\:\Xi\left( \D \omega \right)\\
\nonumber &=& \D^{2}\omega - q [\Xi, \D]\omega\\
\nonumber &=& \D^{2}\omega - q \D \omega,
\end{eqnarray}
for any $\omega \in C^{\infty}(\Pi E)$. Then using the fact that we have a quasi Q-manifold gives
\begin{equation}
\nonumber Q^{2}=0.
\end{equation}
It is clear that $Q(q)=0$ and thus we have a 1-cocycle.
\end{proof}

\begin{proposition}\label{prop:cocycle}
Let $(\Pi E, Q)$ be a Lie algebroid and let $\phi \in C^{\infty}(\Pi E)$ be an odd 1-cocycle, that is $\Xi(\phi)=1$,  $Q(\phi) =0$ and $\widetilde{\phi}=1$. Then
\begin{equation}\nonumber
\left( \Pi E, \D = Q + \phi \: \Xi, q = \phi \right),
\end{equation}
defines a quasi Q-manifold of weigh one, and thus a Jacobi algebroid.
\end{proposition}

\begin{proof}
The conditions on the weights is clear. Then via calculation we obtain
\begin{eqnarray}
\nonumber \D^{2}\omega &=& Q^{2}\omega + \phi \:\Xi \left(\phi \:\Xi( \omega ) \right) + Q \left( \phi\: \Xi( \omega)\right) + \phi \:\Xi \left( Q\omega \right)\\
\nonumber &=& \phi [\Xi,Q]\omega  = \phi \left(Q + \phi \: \Xi  \right)\omega  = \phi \D \omega.
\end{eqnarray}
The 1-cocycle condition implies $\D(\phi)=0$.
\end{proof}

\begin{remark}
The above proposition partially generalises to higher order odd cocycles, one loses the homogeneity in weight of the quasi Q-manifold structure. Thus Lie algebroids in the presence of higher cocycles cannot directly be associated with Jacobi algebroids.
\end{remark}

\begin{theorem}(\textbf{Grabowski--Marmo \cite{Grabowski2001}})
There is a one-to-one correspondence between Jacobi algebroids and Lie algebroids in the presence of an odd 1-cocycle.
\end{theorem}

We must again  remark that everything here is done in the category of supermanifolds and that we have both Grassmann even and odd cocycles. For the classical case where $E \rightarrow M$ is  in the category of pure even classical manifolds 1-cocycles are necessarily odd. Thus the above propositions and theorem include the classical structures.\\

For clarity let us examine the association of a Lie algebroid in the presence of a 1-cocycle with a Jacobi algebroid in natural local  coordinates. It is not hard to see that given $\D$ and $q$ we have

\begin{eqnarray}
Q &=& \xi^{\alpha} Q_{\alpha}^{A} \frac{\partial}{\partial x^{A}} + \frac{1}{2}\left(\xi^{\alpha}\xi^{\beta} Q_{\beta \alpha}^{\gamma} +  (-1)^{\widetilde{\alpha}} 2 \xi^{\alpha}Q_{\alpha} \xi^{\gamma}  \right)\frac{\partial}{\partial \xi^{\gamma}},\\
\nonumber \phi &=& (-1)^{\widetilde{\alpha} +1}\xi^{\alpha}Q_{\alpha}.
\end{eqnarray}

By careful symmetrisation we see that building the Lie algebroid structure on $\Pi E$ associated with a Jacobi algebroid is essentially described by the replacement
\begin{equation}
\nonumber \D  \longrightarrow Q,
\end{equation}
viz
\begin{equation}
\nonumber Q_{\beta \alpha}^{\gamma} \longrightarrow Q_{\beta \alpha}^{\gamma} - (-1)^{\widetilde{\alpha}+ \widetilde{\beta}}\left(\delta_{\alpha}^{\:\: \gamma}Q_{\beta} + (-1)^{(\widetilde{\alpha}+1)(\widetilde{\beta}+1)} Q_{\alpha}\delta_{\beta}^{\:\: \gamma}  \right).
\end{equation}
One can then more-or-less read off the Lie bracket on the sections of $E$ and the anchor map $a : \Gamma(E) \rightarrow \Vect(M)$.  Picking a basis of sections $(s_{\alpha})$ for $\Gamma(E)$ and being intentionally slack with the signs we have

\begin{eqnarray}
\nonumber [s_{\alpha}, s_{\beta}] &=& \pm Q_{\alpha \beta}^{\gamma} s_{\gamma} \pm Q_{\alpha}s_{\beta} \pm s_{\alpha}Q_{\beta},\\
\nonumber a(s_{\alpha}) &=& \pm Q_{\alpha}^{A}\frac{\partial}{\partial x^{A}}.
\end{eqnarray}

Dual to this one can consider the associated Schouten structure which is given by

\begin{equation}
\bar{S} = (-1)^{\widetilde{\alpha}} \pi^{\alpha}Q_{\alpha}^{A}p_{A} + \frac{1}{2}\left( (-1)^{\widetilde{\alpha}+ \widetilde{\beta}}\pi^{\alpha}\pi^{\beta}Q_{\beta \alpha}^{\gamma} + (-1)^{\widetilde{\gamma}}2 \pi^{\alpha} Q_{\alpha}\pi^{\gamma} \right)\eta_{\gamma}.
\end{equation}

Similarly, the 1-cocycle becomes $\bar{\phi} = - \pi^{\alpha}Q_{\alpha}$ and it is not hard to see that

\begin{eqnarray}
\nonumber \{\bar{S}, \bar{S}  \}_{T^{*}(\Pi E^{*})} &=& 0,\\
\nonumber \{ \bar{S} , \bar{\phi} \}_{T^{*}(\Pi E^{*})} &=& 0.
\end{eqnarray}

\begin{corollary}\label{cor:cocyle}
If $\mathfrak{g}$ is a Lie algebra with a distinguished odd 1-cocycle\newline  $\phi \in C^{\infty}(\Pi \mathfrak{g})$ then $\Pi \mathfrak{g}^{*}$ is a (formal) odd Jacobi manifold.
\end{corollary}

In local coordinates we have $Q = \frac{1}{2}\xi^{\alpha}\xi^{\beta}Q_{\beta \alpha}^{\gamma} \frac{\partial}{\partial \xi^{\gamma}} \in \Vect(\Pi \mathfrak{g})$ which \emph{encodes}   the Lie algebra structure on $\mathfrak{g}$. The 1-cocycle is given by $\phi = (-1)^{\widetilde{\alpha}} \xi^{\alpha} Q_{\alpha}$, the sign is picked for convenience. Then the odd Jacobi structure on $\Pi \mathfrak{g}^{*}$ is given by

\begin{eqnarray}
\nonumber S &=& (-1)^{\widetilde{\alpha}+ \widetilde{\beta}}\frac{1}{2} \pi^{\alpha}\pi^{\beta}Q_{\beta \alpha}^{\gamma}\eta_{\gamma}+ (-1)^{\widetilde{\gamma}}\pi^{\alpha}Q_{\alpha} \pi^{\gamma} \eta_{\gamma},\\
\nonumber \Q &=& \pi^{\alpha}Q_{\alpha}.
\end{eqnarray}

The associated odd Jacobi brackets should be thought of  generalisation of the ``Lie--Schouten" bracket on $\Pi \mathfrak{g}^{*}$  \cite{Voronov:2001qf} in the presence of a 1-cocycle. Both these odd brackets are then considered as odd generalisations of the ``Lie--Poisson--Berezin--Kirillov" bracket on $\mathfrak{g}^{*}$.

\begin{corollary}\label{cor:2}
If $(\Pi E, Q)$ is a Lie algebroid, then $\Pi E^{*} \times \mathbbmss{R}^{0|1}$ is a Jacobi algebroid.
\end{corollary}

Let us employ natural local coordinates on $T^{*}(\Pi E^{*} \times \mathbbmss{R}^{0|1})$ which we denote as $(x^{A}, \eta_{\alpha}, \tau, p_{A}, \pi^{\alpha}, \pi)$.  The weight assigned to these extra coordinates is $\w(\tau) =1$ and $\w(\pi)=-1$.  In these local coordinates the weight minus one Jacobi structure is given by

\begin{eqnarray}
\nonumber S&=& (-1)^{\widetilde{\alpha}}\pi^{\alpha}Q_{\alpha}^{A} p_{A} + (-1)^{\widetilde{\alpha} + \widetilde{\beta}} \frac{1}{2}\pi^{\alpha}\pi^{\beta} Q_{\beta \alpha}^{\gamma}\eta_{\gamma} + \pi \pi^{\alpha}\eta_{\alpha},\\
\nonumber \mathcal{Q} &=& - \pi.
\end{eqnarray}

\noindent \textbf{Statement:} extending the fibres of the vector bundle $E \rightarrow M$ underlying a Lie algebroid by $\mathbbmss{R}$ allows one to directly construct a Jacobi algebroid structure on $\Pi E^{*} \times \mathbbmss{R}^{0|1}$.\\

Naturally the proceeding corollary includes Lie algebra as Lie algebroids over a point. Then, if $\mathfrak{g}$ is a Lie algebra one can \emph{extend} the vector space structure to $\mathfrak{g} \times \mathbbmss{R}$. Directly associated with this is the (formal) manifold $\Pi (\mathfrak{g}^{*} \times \mathbbmss{R})$ which comes with an odd Jacobi structure of weight minus one.

\begin{corollary}\label{cor:flatconnection}
Let $M$ be a manifold and $\mathbbmss{A}$ be a closed, odd one-form  (a flat Abelian connection). Then $\Pi TM$ can be made into quasi Q-manifold of weight one, or in other words, $\Pi T^{*}M$ can be considered as a Jacobi algebroid.
\end{corollary}

In natural local coordinates $(x^{A}, dx^{A})$ on $\Pi TM$, the quasi Q-manifold structure is given by

\begin{eqnarray}
\nonumber \D &=& d + \mathbbmss{A} \:\Xi\\
\nonumber &=& dx^{A} \frac{\partial}{\partial x^{A}} + dx^{B}\mathbbmss{A}_{B} \: dx^{A} \frac{\partial}{\partial dx^{A}},\\
\nonumber q &=& \mathbbmss{A} = dx^{B}\mathbbmss{A}_{B}.
\end{eqnarray}

The weight here is simply assigned as $\w(x^{A})=0$ and $\w(dx^{A})=1$. Picking natural local coordinates $(x^{A}, x^{*}_{A}, p_{A}, p^{*}_{A})$ on $T^{*}(\Pi T^{*}M)$ allows us to write the corresponding odd Jacobi structure on $\Pi T^{*}M$ as

\begin{eqnarray}
\nonumber S &=& (-1)^{\widetilde{A}}p_{*}^{A}p_{A} + (-1)^{\widetilde{B}}p_{*}^{B}\mathbbmss{A}_{B}\: p_{*}^{A}x^{*}_{A},\\
\nonumber \mathcal{Q} &=& - p_{*}^{A}\mathbbmss{A}_{A}.
\end{eqnarray}

Note that the first term of the almost Schouten structure is the canonical Schouten structure on the anticotangent bundle.

\subsection{Schoutenization and Lie algebroids}\label{sec:shoutenisationLieAlgebroids}

In this section we show that given arbitrary Jacobi algebroid one can extend the structure via the  Schoutenisation process described earlier to construct a genuine Lie algebroid. Consider the manifold $T^{*}(\Pi E^{*} \times \mathbbmss{R})$ which we equip with local coordinates $(x^{A}, \eta_{\alpha}, t,p_{A}, \pi^{\alpha}, p)$. The weight we assign as:   \\

\begin{tabular}{|l ||l| }
\hline
$ \w(x^{A}) = 0$  & $\w(p_{A}) = 0$ \\
$ \w(\eta_{a})=1$  & $\w(\pi^{a}) =-1$\\
$\w(t)=0$ &   $\w(p)=0$\\
\hline
\end{tabular}\\

\begin{proposition}
Let $(\Pi E^{*}, S, \mathcal{Q})$ be a Jacobi algebroid. Then $\Pi E^{*} \times \mathbbmss{R}$ is a weight minus one Schouten manifold where the Schouten structure is given by
\begin{equation}
\bar{S} = e^{-t} \left( S - \mathcal{Q}p \right).
\end{equation}
\end{proposition}

\begin{proof}
Follows directly from Theorem \ref{th:schoutenization}.  The assignment of the weight follows directly from the definition.
\end{proof}

In natural local coordinates this Schouten structure is given by

\begin{equation}
\bar{S} = e^{-t} \left((-1)^{\widetilde{\alpha}} \pi^{\alpha}Q_{\alpha}^{A}p_{A} + (-1)^{\widetilde{\alpha}+ \widetilde{\beta}}\frac{1}{2}\pi^{\alpha}\pi^{\beta} Q_{\beta \alpha}^{\gamma}\eta_{\gamma} - \pi^{\alpha}Q_{\alpha}p \right).
\end{equation}

We need to understand the vector bundle structure in order to really identify the Lie algebroid structure. Given the weight assigned to the coordinates on $\Pi E^{*} \times \mathbbmss{R}$ the associate underlying (dual) vector bundle structure is $\proj^{*}E \longrightarrow M \times \mathbbmss{R}$. That is the pullback of $E \rightarrow M$ by   $\proj : M \times \mathbbmss{R} \rightarrow M$.

\begin{corollary}
If  $\Pi E^{*}$ has the structure of a Jacobi algebroid then \newline  $\proj^{*}E \longrightarrow M \times \mathbbmss{R}$ is a Lie algebroid.
\end{corollary}

\noindent \textbf{Statement:} given a Jacobi algebroid  structure on $\Pi E^{*}$, one can extend the base space $M$ of the underlying vector bundle $E \longrightarrow M$ by $\mathbbmss{R}$ to directly construct a Lie algebroid.

\subsection{Odd contact manifolds and Jacobi algebroids}\label{sec:odd contact}

In this section we  show that the manifold $M :=  \Pi T^{*}N \times \mathbbmss{R}^{0|1}$ considered as an odd contact manifold provides a canonical example of a Jacobi algebroid than lends itself to the description in terms of odd Jacobi brackets. \\

Let $N$ be a pure even classical manifold of dimension $n$. Consider the manifold $M :=  \Pi T^{*}N \times \mathbbmss{R}^{0|1}$ equipped with natural local coordinates $(x^{a}, x^{*}_{a}, \tau)$. The coordinates $x^{a}$ are even, while the other coordinates $x^{*}_{a}$ and $\tau$ are odd. The dimension of $M$ is $(n|n+1)$. The manifold $M$ comes equipped with an \emph{odd contact one form}, which is the even one form
\begin{equation}
\alpha = d \tau - x^{*}_{a}dx^{a}.
\end{equation}

It was shown in Example \ref{example:oddcontact}. that $M$ is an odd Jacobi manifold with the odd Jacobi structure being

\begin{equation}
S = p^{a}_{*}  \left( p_{a} + x^{*}_{a} \pi\right),  \hspace{30pt} \mathcal{Q} =  - \pi,
\end{equation}
where we have employed natural coordinates $(x^{a}, x^{*}_{a}, \tau, p_{a}, p_{*}^{a} ,\pi)$ on $T^{*}M$. Indeed this odd Jacobi structure is directly equivalent to the odd contact structure. Without details, both the odd contact and odd Jacobi structure on $M$ can be considered as the ``natural superisation" of the classical structures on $ \mathbbmss{R}^{3}$. Note that $\Pi T^{*}N$ comes equipped with a canonical Schouten (odd symplectic) structure, but $\Pi T^{*}N \times \mathbbmss{R}^{0|1}$ comes with a canonical odd Jacobi structure. \\

Let us  attach the weight to the local coordinates on $M$ as:\\

\begin{tabular}{|l ||l| }
\hline
$ \w(x^{a}) = 0$  & $\w(p_{a}) = 0$ \\
$ \w(x^{*}_{a})=1$  & $\w(p_{*}^{a}) =-1$\\
$\w(\tau)=1$ &   $\w(\pi)=-1$\\
\hline
\end{tabular}\\

This weight is the ``natural weight" with respect to the underling  vector bundle structure $T^{*}N \times \mathbbmss{R} \longrightarrow N$. With respect to this weight it is clear that the odd  Jacobi structure on $M := \Pi T^{*}N \times \mathbbmss{R}^{0|1}$ is of weight minus one and we thus have a Jacobi algebroid. \\

Now consider $M^{\star} := \Pi TN \times \mathbbmss{R}^{0|1}$ equipped with natural local coordinates $(x^{a}, \xi^{a}, \eta, p_{a}, \pi_{a}, \theta)$. The canonical double vector bundle morphism $R: T^{*}M \rightarrow T^{*}M^{\star}$ act on the coordinates as $R^{*}(\xi^{a}) = p_{*}^{a}$, $R^{*}(\eta) = \pi$, $R^{*}(\pi_{a})= x^{*}_{a}$ and $R^{*}(\theta)= \tau$. Then we can pull-back the odd Jacobi structure to give

\begin{eqnarray}
\nonumber \hat{S} &=& \xi^{a}(p_{a} + \pi_{a} \eta),\\
\nonumber \hat{\mathcal{Q}} &=& - \eta,
\end{eqnarray}

both of which are now functions on the total space of $M^{\star}$.  Then we can ``undo" the symbol (and after a little reordering) to produce

\begin{eqnarray}
\D &=& \xi^{a}\frac{\partial}{\partial x^{a}} + \eta \xi^{a}\frac{\partial }{\partial \xi^{a}},\\
\nonumber q &=& \eta.
\end{eqnarray}

Direct calculation confirms that $M^{\star} := \Pi TN \times \mathbbmss{R}^{0|1}$ is a quasi Q-manifold.\\

\noindent  \textbf{Statement:}  in light of Proposition \ref{prop:quasiQalgebroid}, $TN \times \mathbbmss{R} \rightarrow N$ is a Lie algebroid in the presence of a 1-cocycle. The de Rham differential on $N$ is the associated homological vector field and the 1-cocycle is identified with the ``odd time".\\

\begin{remark}
As this work was being completed, Mehta \cite{Mehta2011} established a one-to-one correspondence between Jacobi manifolds and degree 1 contact  $NQ$-manifolds. Mehta shows how to interpret the ``Poissonisation" of a Jacobi manifold as the ``symplectification" of the corresponding degree 1 contact $NQ$-manifold. There is no doubt that Mehta's results can be slightly reformulated to sit comfortably with the conventions used here: one would consider ``Shoutenisation" and ``symplectification" of odd contact structures.  This generalises  the  correspondence between Poisson manifolds and degree 1 symplectic $NQ$-manifolds, as established by Roytenberg \cite{Roytenberg:2001}. We also direct the reader to Grabowski \cite{Grabowski2011} who   studies degree 2 contact $NQ$-manifolds as a generalisation of Courant algebroids. The author conjectures that interest in contact structures on super and graded manifolds will continue to grow.
\end{remark}

\section{Concluding remarks}\label{concluding remarks}
In this  paper we defined the notion of an odd Jacobi manifold and examined their basic properties.  In particular,  it was shown that on a supermanifold equipped with an odd Jacobi structure $ J := (S, Q)$ the algebra of smooth functions over the supermanifold $C^{\infty}(M)$ comes with the structure of an odd Jacobi algebra. Furthermore such  the homological vector field $Q$ satisfies a derivation property over the odd Jacobi brackets. \\

However, it remains open as to if interesting or realistic gauge theories exist that require the use of odd Jacobi structures (with $Q\neq 0$) in a generalised Batalin--Vilkovisky formalism. The notion of a gauge system c.f. \cite{Lyakhovich2004} in the context of odd Jacobi manifolds is straight forward.\\

In truth there appears no  applications in physics of even Jacobi structures that cannot simply be restated in terms of Poisson geometry.   That said, it is conceivable that odd contact structures could find quite direct application in  theories with explicit dependency on gauge parameters.\\

The idea of ``odd time" (see Example \ref{example:oddcontact}.) has already been applied in the Batalin--Vilkovisky formalism to get at general and direct solutions of the master equation for a large class of gauge theories, see Dayi   \cite{Dayi1989}. In essence one understands the BRST operator as the partial derivative with respect to the ``odd time" and then  one can formulate the BV formalism in a way akin to classical mechanics. It would be very desirable to properly understand the supergeometry of Dayi's constructions and how this relates  to the work here, in  particular to odd contact manifolds.   The notion of ``odd time" is also essential when constructing flows of odd vector fields. It is certainly expected that odd contact structures are of wider interest than just their relation with odd Jacobi manifolds.  \\

We \emph{defined} Jacobi algebroids  in terms of an odd Jacobi structure on $\Pi E^{*}$ of weight minus one. That is the ``multivector fields" come equipped with an odd Jacobi bracket. For Lie algebroids the bracket between ``multivector fields" is a Schouten bracket, i.e. satisfies a strict Leibniz rule.  \\

This construction was then used to construct a weight one almost homological vector field.  That is the ``differential forms" come equipped with a kind of \emph{deformed} de Rham differential. Importantly we no longer  have a homological vector field as in the case of Lie algebroids, but rather the homological condition is weakened in a very specific way as to provide a quasi Q-manifold structure. As such Jacobi algebroids can be considered as very specific  examples of  \emph{skew algebroids}  \cite{Grabowski1999}, which are a kind of Lie algebroid in which the Jacobi identity is lost. If the corresponding anchor is  a Lie algebra morphism between sections of the vector bundle and vector fields over the base then we have the notion of an \emph{almost Lie algebroid} \cite{Leon2010}. Interest in algebroids without the Jacobi identity comes from nonholonomic mechanics, where skew algebroids provide a general geometric setting.\\

However, via a simple redefinition one can rephrase Jacobi algebroids in terms of Lie algebroids in the presence of a 1-cocycle, which are also known as generalised Lie algebroids. In doing so we recover, maybe up to conventions, the notion of a Jacobi algebroid in the sense of \cite{Grabowski2001,Iglesias2001}. \\

The obvious areas of the present work that require further illumination include:
\begin{itemize}
\item Do Jacobi bialgebroids have an efficient description in terms of a compatible odd Jacobi structure and a quasi Q-structure?
\item Can one define non-linear Jacobi algebroids in terms of odd Jacobi structures over non-negatively graded supermanifolds? Are these naturally related to Voronov's  (\cite{Voronov2010}) non-linear Lie algebroids?
\item Can one develop a theory of higher or homotopy odd Jacobi structures and the related $L_{\infty}$-algebras.
\end{itemize}

\section*{Acknowledgments}
The author would like to thank  Janusz Grabowski, Giuseppe Marmo and Rajan Amit Mehta for their comments and suggestions on  earlier drafts of this work. The author must also than the anonymous referee for their valuable comments.

\newpage
\appendix

\section{ Canonical double vector bundle morphisms}\label{appendix}

For completeness we present the canonical double vector bundle morphisms used in this work. In particular we prove that the morphisms are symplectomorphisms. We describe vector bundles in terms of coordinates on their total spaces and the associated vector bundle automorphisms. Specifically we have:\\

\begin{tabular}{|l||l|}
\hline
$E \longrightarrow M$  &  $E^{*}\longrightarrow M$\\
$(x^{A}, e^{\alpha})  \mapsto (x^{A}) $ & $(x^{A}, e_{\alpha})  \mapsto (x^{A}) $\\
\hline
&\\
$\overline{x}^{A} = \overline{x}^{A}(x)$& $\overline{x}^{A} = \overline{x}^{A}(x)$\\
$\overline{e}^{\alpha} = e^{\beta}T_{\beta}^{\:\: \alpha}(x)$ & $\overline{e}_{\alpha} = \left( T^{-1} \right)_{\alpha}^{\:\: \beta}(x)$\\
\hline
\end{tabular}\\
\vspace{15pt}

Where $T_{\beta}^{\:\: \gamma} \left( T^{-1} \right)_{\gamma}^{\:\: \alpha} = \delta_{\beta}^{\:\: \alpha}$ \emph{etc}. We take  $\widetilde{e^{\alpha}} = \widetilde{e_{\alpha}}= \widetilde{\alpha}$.

Let us employ natural local coordinates:

\vspace{15pt}
\begin{tabular}{|l ||l| }
\hline
$ T^{*}(\Pi E^{*})$ & $(x^{A},\eta_{\alpha}, p_{A}, \pi^{\alpha} )$ \\
$  T^{*}(\Pi E)$  & $(x^{A},\xi^{\alpha}, p_{A}, \pi_{\alpha} )$\\
\hline
\end{tabular}\\

The Grassmann  parities are given by $ \widetilde{x}^{A}= \widetilde{p}_{A}= \widetilde{A}$, $\widetilde{\eta}_{\alpha}= \widetilde{\pi}^{\alpha}= \widetilde{\pi}_{\alpha}= \widetilde{\xi}^{\alpha} =  (\widetilde{\alpha}+1)$. The weights are assigned as $\w(x^{A}) =0$, $\w(\eta_{\alpha}) =1$, $\w(p_{A}) =0$, $\w(\pi^{\alpha}) =-1$ , $\w(\xi^{\alpha}) =-1$, $\w(\pi_{\alpha}) =1$.  The admissible changes of coordinates are:

\vspace{15pt}
\begin{tabular}{|l||l|}
\hline
& \\
$T^{*}(\Pi E^{*})$ &   $\overline{x}^{A}  =  \overline{x}^{A}(x)$, \hspace{5pt} $\overline{\eta}_{\alpha}  =   (T^{-1})_{\alpha}^{\:\: \beta}\eta_{\beta}$,\\
 & $\overline{p}_{A} = \left( \frac{\partial x^{B}}{\partial \overline{x}^{A}} \right)p_{B} + (-1)^{\widetilde{A}(\widetilde{\gamma}+ 1) + \widetilde{\delta}} \pi^{\delta}T_{\delta}^{\:\: \gamma} \left( \frac{\partial (T^{-1})_{\gamma}^{\:\: \alpha}}{\partial \overline{x}^{A}} \right)\eta_{\alpha}$,\\
 & $ \overline{\pi}^{\alpha} = (-1)^{\widetilde{\alpha} + \widetilde{\beta}}\pi^{\beta}T_{\beta}^{\:\: \alpha}$.\\
\hline
&\\
$T^{*}(\Pi E)$ & $ \overline{x}^{A}  =  \overline{x}^{A}(x)$, \hspace{5pt}$\overline{\xi}^{\alpha}  =   \xi^{\beta} T_{\beta}^{\:\: \alpha}$,\\
& $\overline{p}_{A} = \left( \frac{\partial x^{B}}{\partial \overline{x}^{A}} \right)p_{B} + (-1)^{\widetilde{A}(\widetilde{\gamma}+1)} \xi^{\delta}T_{\delta}^{\:\: \gamma} \left(\frac{\partial (T^{-1})_{\gamma}^{\:\: \alpha}}{\partial \overline{x}^{A}}  \right)\pi_{\alpha}$,\\
&  $\overline{\pi}_{\alpha} = (T^{-1})_{\alpha}^{\:\: \beta} \pi_{\beta}$.\\
\hline
\end{tabular}\\

\vspace{15pt}

The canonical double vector bundle morphism $R: T^{*}(\Pi E^{*}) \rightarrow T^{*}(\Pi E )$  are given in local coordinates as

\begin{equation}\nonumber
R^{*}(\pi_{\alpha}) = \eta_{\alpha},  \hspace{35pt} R^{*}(\xi^{\alpha}) = (-1)^{\widetilde{\alpha}}\pi^{\alpha}.
\end{equation}

\noindent\textbf{Lemma A.}\label{lemma1}
The canonical double vector bundle morphism
\begin{equation}\nonumber
R: T^{*}(\Pi E^{*}) \rightarrow T^{*}(\Pi E )
\end{equation}
is a symplectomorphism.

\begin{proof}
The canonical even  symplectic structure on $T^{*}(\Pi E^{*})$ is given by $\omega_{T^{*}(\Pi E^{*})} =  dp_{A}dx^{A} + d\pi^{\alpha}d \eta_{\alpha}$ and on $T^{*}(\Pi E)$ is given by $\omega_{T^{*}(\Pi E)} = dp_{A}dx^{A} + d\pi_{\alpha}d\xi^{\alpha}$. Thus, $R^{*}\omega_{T^{*}(\Pi E )} = \omega_{T^{*}(\Pi E^{*})}$ and we see that $R$ is indeed a symplectomorphism.
\end{proof}

\vfill
\begin{center}Andrew James Bruce\\ \small \emph{Pembrokeshire College},\\
\small \emph{Haverfordwest, Pembrokeshire},\\\small  \emph{SA61 1SZ, UK}\\\small email:\texttt{andrewjamesbruce@googlemail.com}
\end{center}

\end{document}